\documentclass[reqno,11pt]{amsart}

\usepackage{hyperref}
\usepackage{graphicx}
\usepackage{slashed}
\usepackage{amscd}
\usepackage{tensor}
\usepackage{amsmath,amssymb,mathabx,bm}
\usepackage{esint}
\usepackage{dsfont}
\usepackage{accents}
\usepackage[shortlabels]{enumitem}
\usepackage[off]{auto-pst-pdf} 
\usepackage{pst-grad} 
\usepackage{pst-plot} 
\usepackage[mathscr]{eucal}
\numberwithin{equation}{section}
\allowdisplaybreaks[1]

\title{No Events on Closed Causal Curves}

\author[C.F.\ Paganini]{Claudio F. Paganini \\ \\ January 2021}
\address{Fakult\"at f\"ur Mathematik \\ Universit\"at Regensburg \\ D-93040 Regensburg \\ Germany}
\address{Max Planck Institute for Gravitational Physics (Albert Einstein Institute), Am M\"uhlenberg 1, D-14476 Potsdam, Germany}
\email{claudio.paganini@ur.de}

\newtheorem{Def}{Definition}[section]

\newtheorem{Conjecture}{Conjecture}
\newtheorem{Prp}[Def]{Proposition}

\newtheorem{axiom}{Axiom}

\newcommand{\beq}{\begin{equation}}
\newcommand{\eeq}{\end{equation}}
\newcommand{\Proof}{\begin{proof}}
\newcommand{\QED}{\end{proof} \noindent}

\DeclareMathOperator{\tr}{tr}

\newcommand{\Lin}{\text{\rm{BL}}}

\newcommand{\bitem}{\begin{itemize}[leftmargin=2.5em]}
\newcommand{\eitem}{\end{itemize}}

\DeclareFontFamily{OT1}{rsfso}{}
\DeclareFontShape{OT1}{rsfso}{m}{n}{ <-7> rsfso5 <7-10> rsfso7 <10-> rsfso10}{}
\DeclareMathAlphabet{\mycal}{OT1}{rsfso}{m}{n}

\begin{document}

\maketitle

\begin{abstract} 
We introduce the Causal Compatibility Conjecture for the Events, Trees, Histories (ETH) approach to Quantum Theory (QT) in the semi-classical setting. We then prove that under the assumptions of the conjecture, points on closed causal curves are physically indistinguishable in the context of the ETH approach to QT and thus the conjecture implies a compatibility of the causal structures even in presence of closed causal curves. As a consequence of this result there is no observation that could be made by an observer to tell any two points on a closed causal curve apart. We thus conclude that closed causal curves have no physical significance in the context of the ETH approach to QT. This is an indication that time travel will not be possible in a full quantum theory of gravity and thus forever remain a fantasy.  
\end{abstract}

\tableofcontents

\section{Introduction}\label{sec:introduction}
Few theoretical curiosities have inspired the collective human imagination as much as the possibility of time travel, the idea that one might be able to go back in time.  The idea of a time machine was popularized by H. G. Wells' 1895 novel ``The Time Machine'' and the genre endures a lasting popularity. This lasting popularity was aided by the fact that General Relativity admits solutions with closed timelike curves \cite{godel2003example,Thorne:1992gv}. This suggests that time travel might, at least in principle, be physically possible; see for example the lecture notes \cite{shoshany2019lectures} for a pedagogical introduction to the topic. However it is to be remarked that, due to the fact that it has never been observed and the many logical paradoxes associated with time travel, such as the grandfather paradox, time travel is by many considered to be an artefact in the theory. (See for example \cite{lewis1976paradoxes} for a philosophical discussion of the problem). Nevertheless there is a vast scientific literature discussing closed causal curves in various contexts. 

The present work was inspired by Rovelli's recent paper \cite{rovelli2019travel} where he claims that a full thermodynamical treatment along the closed causal curves manages to resolve the apparent paradoxes without resorting to Quantum Theory (QT). 
In the present paper, on the contrary, instead of resorting to thermodynamic and arrow of time type arguments, we will consider closed causal curves in a semi-classical setting similar to the works in \cite{deutsch1991quantum, politzer1994path, hartle1994unitarity, yurtsever1994algebraic}. Deutsch \cite{deutsch1991quantum} studied generalizations based on the ideas of quantum computations in the context of the Many-Worlds interpretation of QT. Politzer \cite{politzer1994path} and Hartle \cite{hartle1994unitarity} studied generalizations both in the density matrix formulation and the sum-over-histories formulation, establishing that the resulting dynamics would be non-linear and non-unitary and would already effect observations prior to the non-chronal region. Finally, Yurtsever \cite{yurtsever1994algebraic} studied the generalization of QT in the presence of Closed Causal Curves based on the algebraic approach to field theory. A different approach to resolve the paradoxes associated with closed timelike curves was taken by Hawking in \cite{hawking1992chronology} where he uses a semi-classical setup to conjecture that a full theory of quantum gravity will prevent closed timelike curves to form in a spacetime evolving from sufficiently regular initial data. 
Our approach in the present paper is closer to the first set of ideas yet it is different in the fact that we start with a novel and consistent formulation of QT \cite{froehlich2019review} which was developed to resolve the measurement problem and investigate what it has to tell us in the semi-classical regime in the presence of closed causal curves. It is worth to mention that the evolution of states in Fr\"ohlich's Events, Trees, Histories (ETH) formulation is already non-linear and non-unitary by default and unitary Schrödinger-Liouville evolution for an open quantum system can be obtained as an approximation \cite{F,frohlich2021time}.

We first collect those definitions from Fr\"ohlich's ETH approach to QT which are relevant to the argument in this paper. In this context we compare the causal structure with the theory of ``Causal Sets''\cite{dowker2005causal}, an other approach to fundamental physics. This comparison shows that closed causal chains can not exist in the relativistic formulation of the ETH approach, essentially by definition of the causal structure. However, it is not entirely clear yet how the classical spacetime emerge in a suitable limit. Therefore one often works in a semi-classical setting with a QT living on a fixed background spacetime. In the present paper we introduce the Causal Compatibility Conjecture, which states, that the definition of the causal structure in the ETH approach to QT is compatible with the causal structure of the underlying spacetime in the semi-classical setup in a particular manner. Then, assuming the Causal Compatibility Conjecture to hold, we show that despite the fact that the ETH approach forbids closed causal chains it is compatible in the semi-classical setting with spacetimes featuring closed causal curves. This is achieved by showing that in this setup no events, according to the definition of the ETH approach, happen along a closed causal curve. Therefore all the points are physically indistinguishable and hence closed causal curves, if they were to exist in the background spacetime, are physically irrelevant. It thus  resolves the paradoxes associated with closed causal curves in a rather surprising way. This is a first practical example of a (pseudo) passive state in the ETH approach. 

\subsubsection*{Overview of the paper}
In Section \ref{sec:preliminaries} we briefly recall the definition of the causal structure in General Relativity. Then in Section \ref{sec:ETH} we first give a informal introduction to the ETH approach to QT before we collect the relevant definitions for the framework and give a short comparison to the theory of ``Causal Sets'' and how this prevents closed causal chains in the full QT. In Section \ref{sec:setup} we discuss the semi-classical setup. In this context we introduce and motivate the Causal Compatibility Conjecture. Finally in Section \ref{sec:ccc} we prove Proposition \ref{equal} which is the main result of this paper.

 \section{Background}\label{sec:preliminaries} 
In this section we collect the definitions of the causal structure in General Relativity and the structures of the ETH approach to QT which are relevant for the argument in the paper. 
\subsection{Causal Structure in General Relativity}
As it is central for the argument in this paper we shortly recall the definition for the causal structure in General Relativity. The central postulate of General Relativity is a $4$-dimensional time orientable Lorentzian manifold equipped with a metric (${M}$,$g$). 
Here we choose the signature of $g$ to be ($-$,$+$,$+$,$+$). A parametrized curve $\gamma$ at a point $p$ in the manifold ${M}$ can either be time-, space-, or lightlike/null depending on whether 
\begin{equation}
\epsilon(p) = \left.-g_{\rho \sigma}\dot \gamma^\rho \dot \gamma^\sigma \right|_p
\label{eq:timespacenull}
\end{equation}
is positive, negative, or zero. Here $\dot \gamma^\rho $ is the tangent vector of the curve $\gamma$ at the point $p$. This definition is then used to introduce a causal structure on the spacetime (${M}$,$g$). An event at one point on the manifold can influence events at another point on the manifold if there exists a causal, future directed curve from the first to the second point. A causal curve is one that is everywhere timelike or null. The causal future of a set of points denoted by $S$ on a manifold ${M}$ is denoted by $J^+(S)$. It consists of all points in ${M}$, that can be reached from a point in $S$ by a future directed, causal curve. The causal past $J^- (S)$ is  defined analogous to its future counterparts. A closed causal curve $\gamma$ is a closed curve who's tangent vector is everywhere timelike or null. Thus by travelling along a closed causal curve one can return to the same event in spacetime while always moving forward in time locally.

\subsection{``Events Trees Histories'' Approach to Quantum Theory}\label{sec:ETH} 
Given the fact that the ETH approach to QT is quite new and unfamiliar for most readers we will start of this section with an informal summary of the ideas behind the framework and how it relates to more established formulations.

To contextualize the ideas at the heart of the ETH approach to QT we start with a short discussion about the Copenhagen interpretation as it is probably the approach to QT known best by most readers. A system evolves according to the Schrödinger equation (i.e. it evolves unitarily) as long as it is left alone. To make a measurement, one specifies a physical observable which one wants to measure and one executes an operation that allows to access that particular information about the system one is interested in. Any observable (i.e. physical property of the system) is represented by a self-adjoint operator and the measurement result one will obtain is one of the eigenvalues of this operator. The state of the quantum system tells an observer with which probability they will measure a particular eigenvalue. According to observations in experiments with repeated measurements we know that to obtain the correct prediction for subsequent measurements, the state of the system after a successful measurement has to be projected to the eigenspace corresponding to the eigenvalue obtained as a result in said measurement. This has some issues. In particular the fact that you can choose when and what observable you want to measure and that your measurement, hence your actions, seem to influence the evolution of the system. The fact that the measurement changes the evolution of the physical system, as during measurement the system under observation does not evolve according to the Schrödinger equation (i.e. not unitary) leaves us with a curious situation: we now have two sets of rules for the evolution of a quantum system, one set of rules that applies during measurements and another set that applies the rest of the time. This leaves the underlying reality in question.

Now, the Many Worlds interpretation of QT, as advocated by Deutsch \cite{deutsch1991quantum} for example, resolves these issues by claiming that the collapse of the wave-function does not exist, but all evolution is unitary. However, the wave-function splits into different branches that decohere and an observer just happens to live in a branch that observes a particular measurement outcome. All the alternative branches are as real as the one the observer sees and the wave-function gives the probability that they experience the branch of the wave-function they believe to live in. However, it is unclear when and how exactly a split occurs and how the splitting procedure is compatible with relativity. 

The ETH (Events, Trees, Histories) approach to quantum mechanics resolves these issues in a novel and interesting manner. In their recent paper \cite{frohlich2021time} Fr\"ohlich and Pizzo state about the ETH approach to QT that: 

\textit{``One might argue that `Principle of Diminishing Potentialitie(PDP)' and the Collapse Postulate provide a mathematically precise version of the Many-Worlds Interpretation of QM. However, in the ETH-Approach, there is no reason, whatsoever, to imagine that many alternative worlds actually exist!"}

Contrary to Many Worlds which is usually discussed in the Schrödinger picture, where the operators corresponding to observables are fix, and the state evolves according to the Sch\"odinger equation (unitarily), the ETH approach to QT is formulated in the Heisenberg picture, where the operators evolve by unitary conjugation while the state remains fixed. The heart of the ETH formalism is a sharp definition of an event. At every time-step (for simplicity assuming discrete time) the formalism selects a self-adjoint operator according to a set of rules which we will explain in detail below. If that operator has more than one distinct eigenvalue, the (otherwise fix) state is projected onto the eigenspace corresponding to one of the eigenvalues, with a probability given by Born's rule.  Therefore one can think of Events as a sort of “natural” measurement in the Copenhagen sense. With this you arrive at a stochastic evolution for your open local quantum system. An Event basically corresponds to a “split” in the context of the Many Worlds interpretation, except that here we have a sharp definition when it occurs and what it implies. Furthermore the ETH approach is compatible with relativity as we will discuss below.

It is worth noting, that ETH approach to QT deals with \textit{open} and \textit{isolated} systems. The fact that the system is open is crucial, because a closed isolated system does \textit{not} admit any events and thus evolves unitarily. In the recent paper \cite{frohlich2021time} Fröhlich and Pizzo show, that if you take a local non relativistic quantum system (LQS) and couple it to a mock electromagnetic field you can interpolate the evolution of the LQS between stochastic evolution, in the case of strong coupling between the LQS and the mock electromagnetic field, and (almost) unitary evolution in case of weak coupling between the LQS and the mock electromagnetic field. They then proceed to discuss the case of an actual experiment where the LQS splits into a probe and a measuring device. The measuring device is assumed to always be strongly coupled to the electromagnetic field (think thermal radiation emitted by any macroscopic object). Assuming the probe is weakly coupled to the electromagnetic field, then its evolution is (almost) unitary until it interacts with the measuring device (i.e. until you measure it). This is in accordance with our experience in actual experiments. This concludes the informal introduction of the ETH approach to QT and we will no proceed to introduce the detailed mathematical structures.  

Here we only collect those structures and definitions of the ETH approach that are relevant to the argument in the later sections. For a reader unfamiliar with the approach, the recent review by Fr\"ohlich \cite{froehlich2019review} or the discussion in \cite{cfseth} is a good place to start. For technical details on the non-relativistic formulation see \cite{frohlich2015math}.
Here, we limit the discussion to the relativistic formulation and closely follow \cite{froehlich2019relativistic} and \cite{cfseth} but restricting the exposition to a minimum of what is needed to understand the arguments in the present paper. We will use the occasion to point out similarities between the relativistic formulation of the ETH approach and the theory of ``Causal Sets'' \cite{dowker2005causal} and use this to argue that the full QT prohibits the existence of closed causal chains. This leads to a tension in the causal structures in the semi-classical regime where the underlying spacetime features closed causal curves. We will ultimately resolve this tension in Section \ref{sec:ccc}.  

In the ETH formulation of relativistic QT, a model of an isolated open physical system $S$ is defined by specifying the following data: 
\[ 
S=\big\{\mathcal{M}, \mathcal{E}, \mathcal{H}, \big\{\mathcal{E}_{P}\big\}_{P\in \mathcal{M}},\succ \big\}\,, \]
where $\mathcal{M}$ is a model of spacetime, $\mathcal{E}$ is a $C^{*}$-algebra represented on a Hilbert space 
$\mathcal{H}$, $\big\{ \mathcal{E}_{P} \big\}_{P \in \mathcal{M}}$ is a family of von Neumann algebras associated with every point in $M$ and $\succ$ is the relation  on $M$ induced by the ``Principle of Diminishing Potentialities'' for timelike separated points which is given by the following definition.

\begin{Def}\label{def:timelikeETH}
A spacetime point $P'$ is in the future of a spacetime point $P$, written as $P'\succ P$ (or, equivalently, $P$ is in the past of $P'$, written as $P\prec P'$\,) if
\begin{equation}\label{PDP2}
\boxed{
\,\mathcal{E}_{P'} \subsetneqq \mathcal{E}_{P}\,, \quad \mathcal{E}'_{P'} \cap \mathcal{E}_{P} \,\,\text{is an } \infty-\text{dim. non-commutative algebra}\,
}
\end{equation}
Here $\mathcal{E}'_{P'}:=\{x \in \mathcal{E}|\, \,\,  [x,y]=0, \,\,\, \forall y \in \mathcal{E}_{P'} \} $ is the commutator of the algebra $\mathcal{E}_{P'}$.
\end{Def}

In fact, this definition holds as a theorem in an axiomatic formulation of quantum electrodynamics in four-dimensional Minkowski space proposed by
 Buchholz and Roberts \cite{buchholz2014new} (This is in fact the key motivation for the Causal Compatibility Conjecture which will be introduced in Section \ref{sec:setup}) 
\begin{Def}\label{def:spacelikeETH}
If a spacetime point $P'$ is neither in the future of a spacetime point $P$ nor in the past of $P$ we say that $P$ and $P'$ are \textit{space-like separated}, written as
$P$\, {\Large$\bigtimes$} \,$P'$. \hspace{6.4cm} 
\end{Def}
\begin{figure}[h]
\centering
\includegraphics[width= 60mm]{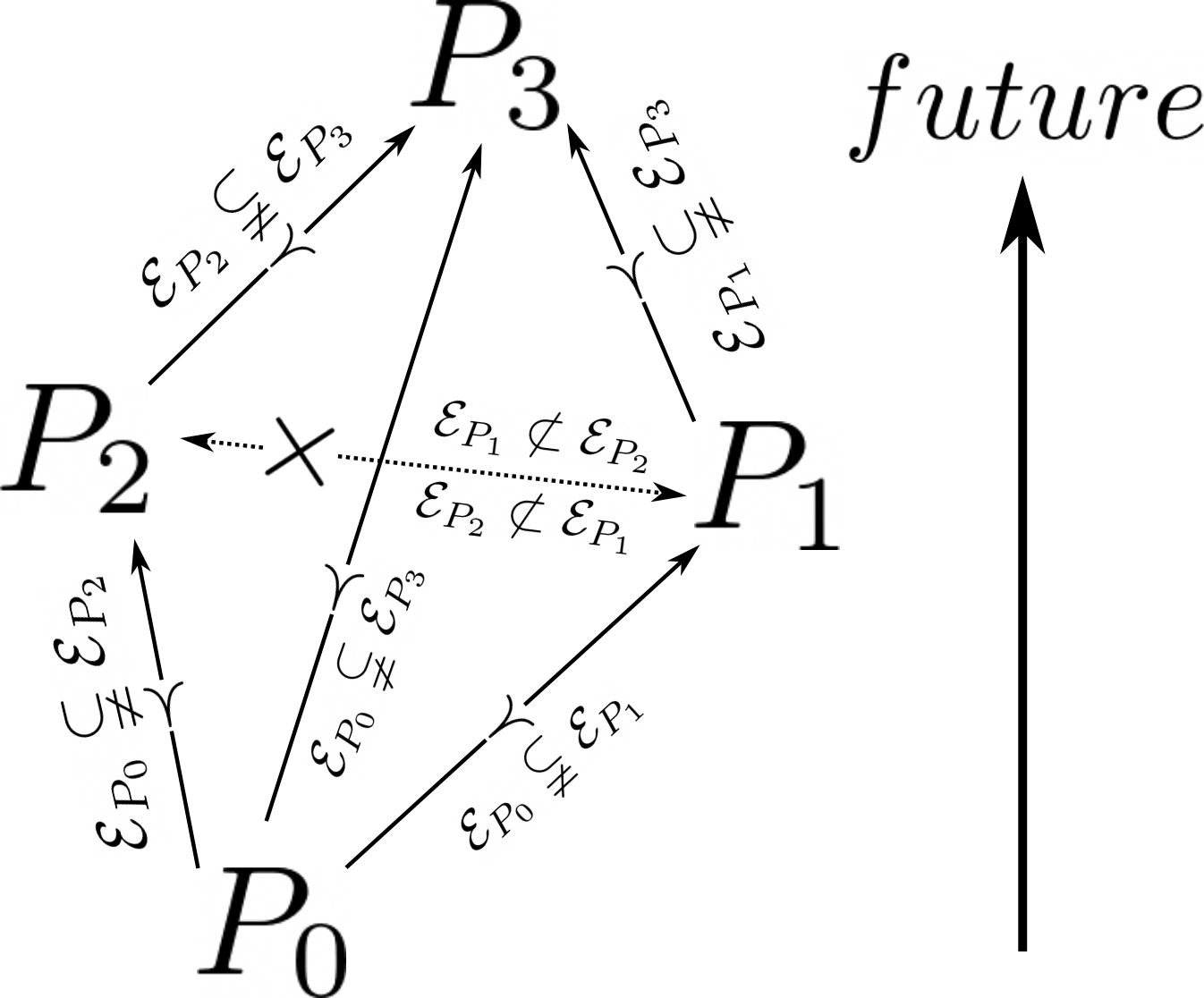}
\caption{Schematic representation of the causal structure between four points and the relations between their corresponding algebras. The points $P_1$ and $P_2$ are spacelike separated, all other relations are causal.}
\label{fig:causal}
\end{figure}
See Figure \ref{fig:causal} for a schematic depiction of the causal structure.\\
At this point it is interesting to remark that the strict partial order defined on $\mathcal{M}$ by Definition \ref{def:timelikeETH} is compatible with the first two axioms of Causal Set Theory \cite{dowker2005causal}. In Causal Set Theory a causal set is defined to be a locally finite partially ordered set. This means a set $C$ together with a relation $\prec$, called “precedes”, which satisfy the following axioms:
\begin{enumerate}[(a)]
    \item if $x \prec y$ and $y \prec z$ then $ x \prec z$, $\forall x,y,z \in C$ (transitivity);
    \item if $x \prec y$ and $y \prec x$ then $x = y$  $\forall x,y \in C$ (non-circularity);
    \item for any pair of fixed elements $x$ and $z$ of $C$, the set $\{y\,|\,\, x \prec y \prec z\}$ of elements lying between $x$ and $z$ is finite.
\end{enumerate}
If we consider the spacetime $\mathcal{M}$ as our set of interest, then it is clear that the relation in Definition \ref{def:timelikeETH} satisfies both axiom (a) and axiom (b) of Causal Set Theory. It then follows immediately from axiom (b) that no closed causal chain of points in $\mathcal{M}$ can exist. Before we can talk about axiom (c) in the context of the ETH approach we need to define under which conditions an actual event happens. 

We denote by $\Omega$ the density matrix on $\mathcal{H}$ representing the actual state of the system $S$. To denote the expectation value of a self-adjoint operator $X$ (which represents a physical characteristic of the system $S$)  in the state $\omega$ determined by $\Omega$, we use the notation 
$$\omega(X):= \tr(\Omega\,X), \qquad \forall X=X^* \in\Lin(\mathcal{H}), $$
where $\Lin(\mathcal{H})$ is the set of bounded linear operators on $\mathcal{H}$. This definition of $\omega$ is then extended to a functional on all of $\Lin(\mathcal{H})$. 
We next define the centralizer and the center of an algebra for a given state. 
\begin{Def}[Centralizer]
Given a $^{*}$-algebra $\mathcal{A}$ and a state $\omega$ on $\mathcal{A}$, the centralizer, $\mathcal{C}_{\omega}(\mathcal{A})$, of the state $\omega$ is the subalgebra of $\mathcal{A}$ spanned by all operators $Y$ in $\mathcal{A}$ with the property that 
$$\omega([Y, X]) =0, \qquad \forall X \in \mathcal{A},$$
i.e.
$$\mathcal{C}_{\omega}(\mathcal{A}):=\left\{Y\in \mathcal{A}|\,\,\omega([Y, X]) =0, \,\,\forall X \in \mathcal{A}\right\} .$$
\end{Def}
Note that with the centralizer defined in this way the state $\omega$ is a (normalized) trace on this subalgebra of $\mathcal{A} $.
\begin{Def}[Center of the Centralizer]
The \textit{center} of the centralizer, denoted by $\mathcal{Z}_{\omega}(\mathcal{A})$, is the abelian subalgebra of the centralizer  consisting of all operators in $\mathcal{C}_{\omega}(\mathcal{A})$ which commute with all other operators in 
$\mathcal{C}_{\omega}(\mathcal{A})$, i.e.
$$\mathcal{Z}_{\omega}(\mathcal{A}):= \big\{ Y \in \mathcal{C}_{\omega}(\mathcal{A})| \,\,[Y,X]=0\,\, \forall X \in \mathcal{C}_{\omega}(\mathcal{A}) \big\}. $$
\end{Def}
In the following we will denote with $\omega_P$ the state of the quantum system right before the spacetime point $P$. Further down we will discuss how this state is determined. With these definitions at hand an event is defined in the following way. 
 \begin{Def}[Event]\label{ETHevent}
 An event $\lbrace \pi_{\xi}, \xi \in \mathfrak{X} \rbrace \subset \mathcal{E}_{P}$, with 
$\lbrace \pi_{\xi}, \xi \in \mathfrak{X} \rbrace$ not contained in $\mathcal{E}_{P'}, \,\text{  for  }\, P'\prec P$, 
starts happening in the future of $P$ if $\mathcal{Z}_{\omega_P}(\mathcal{E}_P)$ is non-trivial,\footnote{The algebra $\mathcal{Z}_{\omega_P}(\mathcal{E}_P)$ is an abelian von Neumann algebra. On a separable Hilbert space, it is generated by a single self-adjoint operator $G$, whose spectral decomposition yields the projections,
$\lbrace \pi_{\xi}, \xi \in \mathfrak{X} \rbrace$, describing a potential event.}
\begin{equation}\label{act-event}
\lbrace \pi_{\xi}, \xi \in \mathfrak{X} \rbrace \,\,\text{generates   }\,\, \mathcal{Z}_{\omega_{P}}\big(\mathcal{E}_{P}\big),
\end{equation}
and 
\[ 
\omega_{P}(\pi_{\xi_{j}}) \,\,\text{ is \textit{strictly positive}},\, \,\, \xi_j \in \mathfrak{X}, \,\, j=1,2, \dots, n\,, \]
for some \, $n\geq 2$.
 \end{Def}
Here $ \mathfrak{X}$ denotes a Hausdorff topological space	of orthogonal projections, $\pi_{\xi}$, on $\mathcal{H}$ with the properties
	\begin{equation}\label{partition-unity}
\pi_{\xi}\cdot \pi_{\eta} = \delta_{\xi \eta} \pi_{\xi}\ \,\,\mbox{for all } \xi, \eta \text{ in }\mathfrak{X}, \quad\mbox{and}\quad \sum_{\xi \in \mathfrak{X}} \pi_{\xi} = {\bf{1}}\,.
	\end{equation}
\begin{figure}[h!]
\centering
\includegraphics[width= 140mm]{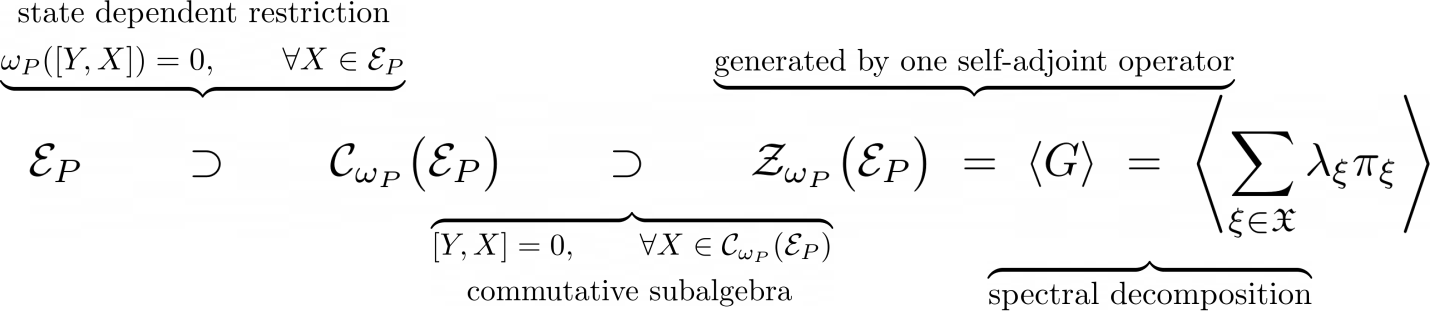}
\caption{Collection of the different algebraic structures required for the definition of an event. Here $G$ is the the single self-adjoint operator which generates the abelian von Neumann algebra $\mathcal{Z}_{\omega_P}(\mathcal{E}_P)$.}
\label{fig:eventdef}
\end{figure}
Given the complexity of the definitions involved in the definition of an event, Figure \ref{fig:eventdef} provides a collection of the different structures involved for an easier orientation how these structures fit together.\\
It is expected that events usually have a finite extent in spacetime (see~\cite{froehlich2019relativistic}).
This implies that the operators 
$\big\{ \pi_{\xi} \vert \xi \in \mathcal{X}\big\}$ representing a potential event in the future of the point $P$ would be localized in a compact region of spacetime contained in the future of $P$ (the future light-cone with apex at $P$). Accordingly, if we only look at the events in $\mathcal{M}_E(\omega)\subset\mathcal{M}$, hence if we eliminate all points from $\mathcal{M}$ except those for which, given a state $\omega$, an event begins in their immediate future, then it seems reasonable to expect, that $\mathcal{M}_E(\omega)$ satisfies axiom (c) of Causal Set Theory. 
However at this point we need to emphasize that the question regarding the ``size'' of events is an open problem. Though it seems plausible to assume that an event is no smaller than the Planck scale. As this is still an open question however, for simplicity we will assume in the following, that events are point like. 

Suppose at a spacetime point $P$ an event occurs according to definition \ref{act-event}. Now we discuss how this effects the further evolution of the system $S$ under consideration. Let $\omega_{P}$ be the state of an isolated system $S$ right before the spacetime point $P$. Let us suppose that an event $\lbrace \pi_{\xi}, \xi \in \mathcal{X} \rbrace$ generating $\mathcal{Z}_{\omega_P}(\mathcal{E}_{P})$ happens at $P$. The ETH approach requires the following Axiom.
\begin{axiom}\label{collapsaxiom} The actual state  of the system $S$ right after the event $\lbrace \pi_{\xi}, \xi \in \mathfrak{X} \rbrace$ starting at $P$ has happened is given by one of the states
\[ 
\omega_{P, \xi_{*}}(\cdot):=[\omega_{P}(\pi_{\xi_{*}})]^{-1}\,\omega_{P}\big(\pi_{\xi_{*}} (\cdot) \pi_{\xi_{*}}\big)\,, \]
for some $\xi_{*} \in \mathfrak{X}$ with $\omega_{P}(\pi_{\xi_{*}})>0$. The probability for the system $S$ to be found in the state $\omega_{P,\xi_{*}}$ right after the event 
$\lbrace \pi_{\xi}, \xi \in \mathfrak{X} \rbrace$ starting at $P$ happened is given by Born's Rule, i.e., by
\begin{equation}\label{Born}
\text{\rm{prob}}\{\xi_{*}, P\} = \omega_{P}(\pi_{\xi_{*}}).
\end{equation}
\end{axiom}
Note that with this assumption, whenever an event happened the state is projected onto the eigenspace of the outcome of the event, the dynamics of the state in the ETH approach in the non-relativistic setting \cite{frohlich2015math} is essentially the same as that in quantum trajectory theory (QTT, see eg. \cite{brun2002simple}) with the difference that in QTT one has to invoke ``repeated measurements'' while in the ETH approach the projection, by Axiom \ref{collapsaxiom}, is naturally part of the dynamics, and there is no need to invoke an observer performing measurements. In fact it turns out to be the other way around in the context of the ETH approach to QT: a successful measurement is the recording of an event happening in the physical system. An instrument in a laboratory is a physical system that generates events in a predictable manner. 

Finally we consider the evolution of the state in the relativistic formulation of the ETH approach to QT in a  spacetime $\mathcal{M}$.  Let $P,P'$ be points in $\mathcal{M}$, let $\mathcal{Z}_{\omega_P}$ denote the center of the centralizer of the state $\omega_P$ on the algebra $\mathcal{E}_P$, which describes the event $\big\{\pi^{P}_{\xi} \vert \xi \in \mathcal{X}^{P} \big\}$ happening in the future of $P$, and let 
$\mathcal{Z}_{\omega_{P'}}$ be the algebra describing the event happening in the future of the point $P'$.
We now introduce the following axiom.

\begin{axiom}\label{axiom2}
(Events in the future of space-like separated points commute): Let $P\, \bigtimes \, P'$. Then all operators in $\mathcal{Z}_{\omega_P}$ commute with all operators in $\mathcal{Z}_{\omega_{P'}}$. In particular,
$$\hspace{2.2cm} \big[\pi^{P}_{\xi}, \pi^{P'}_{\eta}\big] =0, \,\,\,\forall\, \xi \in \mathcal{X}^{P}  \text{ and all }\, \eta \in \mathcal{X}^{P'}.$$
\end{axiom}

Following~\cite{froehlich2019relativistic}, this axiom may be one reflection of what people sometimes interpret as the fundamental non-locality of quantum theory: projection operators representing events in the future of two space-like separated points $P$ and $P'$ in spacetime are constrained to commute with each other! This can be seen as the equivalent condition of what Dowker calls ``discrete general covariance'' in the context of Causal Set Theory \cite{dowker2005causal} namely that the probability of growing a particular finite partial causet does not depend on the order in which the elements are ``born''. As we will discuss next, in the context of the ETH approach Axiom \ref{axiom2} guarantees that the probability of a particular event happening does not depend on the ``order'' in which spacelike separated events occur.

To be able to formulate the evolution of the system $S$ in terms of initial data, we consider the evolution of $S$ in a finite slice $\mathfrak{F}$ of the spacetime $\mathcal{M}$.  Assume that $\mathfrak{F}$ contains a space-like hypersurface 
$\Sigma_{0}$ on which there exists an initial state $\omega_{\Sigma_0}$ (for a detailed construction see \cite{froehlich2019relativistic}). Denote by $V^{-}_{P}(\mathfrak{F})$ all points in $\mathfrak{F}$ that lie in the past of $P$. Let $\mathfrak{I}( \mathfrak{F})$ be a set of indices labelling the points in $V^{-}_{P}(\mathfrak{F})$ in whose future events happen; it is here assumed to be countable.
Then 
$\big\{P_\iota \vert \iota \in \mathfrak{I}(\mathfrak{F})\big\}$ denotes the subset of points in $V^{-}_{P}(\mathfrak{F})$ in whose future events happen, and let 
$$\big\{ \pi_{\xi_\iota}^{P_\iota} \vert \iota \in \mathfrak{I}(\mathfrak{F})\big\} \subset \mathcal{E}_{\Sigma_0}$$ 
be the actual events that happen in the future of the points $P_\iota\,, \iota \in \mathfrak{I}(\mathfrak{F})$. Here $\mathcal{E}_{\Sigma_0}$ denotes the union of the algebras of pontentialities of all points in $\Sigma_0$. 
Fr\"ohlich \cite{froehlich2019relativistic} then defines a so-called \textit{``History Operator''} 
\begin{equation}\label{History}
H\big(V^{-}_{P}(\mathfrak{F})\big):= \vec{\Pi}_{\iota \in \mathfrak{I}(\mathfrak{F})}\, \pi_{\xi_\iota}^{P_\iota}\,,
\end{equation}
where the ordering in the product $\Vec{\Pi}$ is such that a factor $\pi_{\xi_\kappa}^{P_\kappa}$ corresponding to a point $P_{\kappa}$ stands to the right of a factor $\pi_{\xi_\iota}^{P_\iota}$ corresponding to a point $P_{\iota}$ if and only if $P_{\kappa} \prec P_{\iota}$ (i.e., if $P_{\kappa}$ is in the past of $P_{\iota}$). But if $P_{\iota} \bigtimes P_{\kappa}$, i.e., if $P_{\iota}$ and $P_{\kappa}$ are space-like separated, then the order of the two factors is irrelevant thanks to Axiom 2! The History Operator thus holds a record of all events that happened in the past of the point $P$.  In Figure \ref{fig:history} we give an schematic example of a history operator in a simple situation. 
\begin{figure}[t!]
\centering
\includegraphics[width= 80mm]{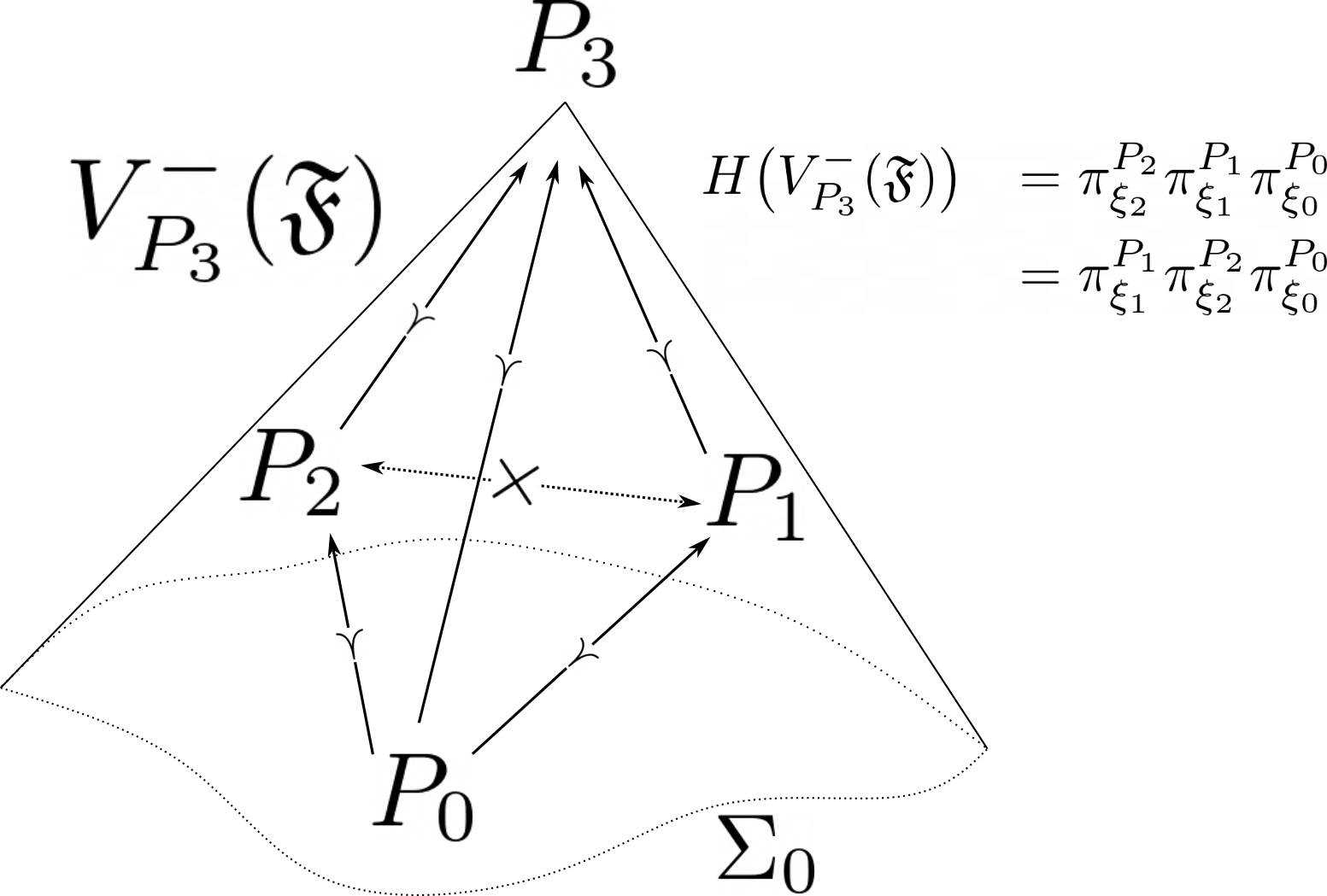}
\caption{A simple schematic example of a history operator for a point $P_3$ and initial data surface $\Sigma_0$, where only three events happened in the recent past $V^{-}_{P_3}(\mathfrak{F})$ of $P_3$   at the points $P_0$, $P_1$ and $P_2$. We assume the same causal relations between the four points as in Figure \ref{fig:causal}.}
\label{fig:history}
\end{figure}
The state on the algebra $\mathcal{E}_{P}$ relevant for making predictions about events happening in the future of $P$ is then given by
\begin{align}\label{state-propagation}
\omega_{P}(X) &\equiv \omega_{P}^{\mathfrak{F}}\big(X \big) \\
&= \big[\mathcal{N}_{P}^{\,\mathfrak{F}}\, \big]^{-1} 
\omega_{\Sigma_0}\big(H(V^{-}_{P}(\mathfrak{F}))^{*} \,X\, H(V^{-}_{P}(\mathfrak{F}))\big)\,, \nonumber
\end{align}
with $X\in \mathcal{E}_{P}$ and where the normalization factor $\mathcal{N}_{P}^{\,\mathfrak{F}}$ is given by
\[ 
\mathcal{N}_{P}^{\,\mathfrak{F}}= \omega_{\Sigma_0}\big(H(V^{-}_{P}(\mathfrak{F}))^{*} \cdot H(V^{-}_{P}(\mathfrak{F}))\big)\,. \]

The quantities $\mathcal{N}_{P}^{\,\mathfrak{F}}$ can be used to equip the tree-like space (what Fr\"ohlich 
in~\cite{froehlich2019review} refers to as the ``non-commutative spectrum'' of $S$) of all possible histories of events in the future of $\Sigma_0$ with a probability measure. It would be interesting to investigate how the ``non-commutative spectrum'' in the ETH approach relates to the ``law of growth'' in the context of Causal Set Theory \cite{dowker2005causal}. However this is beyond the scope of this paper.

This are all the details one needs to know about the ETH approach with regards to the following considerations. 
\section{The Semi-Classical Set-Up}\label{sec:setup}

For the key argument in this paper we consider a semi-classical setup in the following sense. We assume there exists a classical time orientable manifold $(M,g)$ as a background on which we study a quantum system which features massless modes and follows the ETH dynamics. This is similar to the setup used in the proof of the ``Principle of Diminishing Potentialities'' by  Buchholz and Roberts \cite{buchholz2014new} in the context of four-dimensional Minkowski space. We conjecture that this result extends to quantum theories with massless modes on generic well behaved four dimensional spacetimes and thus the causal structure provided by the ETH approach to QT to be compatible with the causal structure of the background manifold. In the following we will use this conjecture as an assumption to discuss the situation of spacetimes with closed causal curves.  
\begin{Conjecture}[Causal Compatibility]\label{future2}Let $(M,g)$ be a suitably well behaved spacetime manifold. Consider a quantum theory containing massless modes. Then we have that for two points $x,y\in M$, $J^+(x)\subset J^+ (y)$ implies $\mathcal{E}_x\subset \mathcal{E}_y$.
\end{Conjecture}
The Conjecture \ref{future2} is motivated in three ways.  First, by the result of Buchholz and Roberts \cite{buchholz2014new} in Minkowski space. Second, this is a necessary assumption for the causal structure of the ETH approach to QT to be compatible with the causal structure of the background spacetime. Classically $J^+(x)\subset J^+ (y)$ implies that $x$ is in the future of $y$. In the ETH approach to QT $\mathcal{E}_x\subset \mathcal{E}_y$ is compatible with the statement that $x$ is in the future of $y$. The third motivation comes from a recent paper \cite{cfseth} comparing the ETH approach to QT with the Causal Fermion Systems theory (see \cite{website} for an introduction). In the CFS setting the relevant future algebras are constructed explicitly from the future light cone. In PDE language the conjecture translates to the question, whether there exist solutions to the corresponding field equations which have a source function that is compactly supported in the causal diamond between $y$ and $x$ and which are equal to zero in the future of $x$. In Minkowski space, and spacetimes (locally) conformal to it, this is guaranteed for every solution by Huygens principle \cite{friedlander1976wave}. In general spacetimes Huygens principle breaks down due to backscattering, which has been particularly well studied in the context of black hole spacetimes (see e.g.\cite{hod1999high, hod2000radiative}. Due to backscattering, for generic solutions one can at most have polynomial decay in time. This is the main obstacle for the causal compatibility conjecture. There is hope however, that particular solutions exist, that satisfy the conditions mentioned above. On the one hand, in the high frequency approximation it is known  \cite{sbierski2015characterisation, ralston1969solutions} that for any finite time there exist solutions to the wave equation focused along a chosen null geodesics. Hence, if there exists a null geodesics $\gamma\in J^+(y)\backslash J^+(x)$ then for every finite time there exists a solution that is almost zero in $J^+(x)$. However, this is of course not good enough, but on the other hand, we can consider the backwards characteristic initial value problem, with initial data on the boundary of $ J^+(x)$ set to zero and chosen arbitrarily on the intersection of $J^+(x)\backslash J^+(y)$ with future null infinity (or $\partial M$ respectively\footnote{In the context of black hole spacetimes and cosmological spacetimes it is unclear, whether one should consider the entire future of $x$ or only the part intersecting with the domain of outer communication. Though likely this will not change the answer to the conjecture.}). It is unclear however whether the intersection of this set of solutions with those we are interested in is non-empty.

\begin{Prp}
\label{past}
Assuming the Causal Compatibility Conjecture to hold, $x,y\in M$, then if $J^-(y)\subset J^-(x)$ there exists an admissible ordering of the history operators $H\big(V^{-}_{x}(\mathfrak{F})\big)$, $H\big(V^{-}_{y}(\mathfrak{F})\big)$ such that $H\big(V^{-}_{y}(\mathfrak{F})\big)$ is a factor of $H\big(V^{-}_{x}(\mathfrak{F})\big)$. In particular we have 
\begin{equation}
    H\big(V^{-}_{x}(\mathfrak{F})\big) =  \vec{\Pi}_{\iota \in \left[\mathfrak{I}_x(\mathfrak{F})\backslash  \mathfrak{I}_y(\mathfrak{F})\right] }\, \pi_{\xi_\iota}^{P_\iota}\,\, \cdot H\big(V^{-}_{y}(\mathfrak{F})\big).
\end{equation}
\end{Prp}

\begin{proof}
 The proposition is a simple application of Axiom \ref{axiom2} combined with the Causal Compatibility Conjecture \ref{future2}. First note that points in $J^-(x)\backslash J^-(y)$ are never in the past of points in $J^-(y)$ but either in the future of them or spacelike separated. The Causal Compatibility Conjecture together with Axiom \ref{axiom2} imply that we either have to (if in the future) or can (if spacelike separated) write all the projections corresponding to events that happen in the spacetime region $ J^-(x)\backslash J^-(y)$ to the left of those projections corresponding to events that happen in $J^-(y)$.
\end{proof}

\section{Closed Causal Curves in the Semi-Classical Regime}\label{sec:ccc}
We will now discuss the physical nature of closed causal curves in the above defined semi-classical regime and we will show that, assuming the Causal Compatibility Conjecture to hold, the causal structure of the relativistic formulation of the ETH approach, despite not allowing for closed causal chains, is indeed compatible with the causal structure of the spacetime even in the presence of closed causal curves. 

Because we aim to discuss closed causal curves we will deviate slightly from the formulation of the relativistic QT in the ETH approach in \cite{froehlich2019relativistic} in that we will assume a general state $\omega$ on the global algebra $\mathcal{E}$ to be a normalized, positive linear functional on $\mathcal{E}$. This is in contrast to Section \ref{sec:ETH} where we defined a state $\omega_{\Sigma_0}$ with respect to a particular Cauchy surface. Accordingly we replace $\mathfrak{F}$ with $\mathcal{M}$ in the definition of the history operator and we assume the events in the past of every point to be countable. This setup is a simplification to avoid a technical discussion about the choice of a relevant initial data surface in a spacetime with closed causal curves. We can then prove the following Proposition. 
\begin{Prp}\label{equal}
Let $x$ and $y$ be two arbitrary distinct points on a closed causal curve. Given Assumption \ref{future2} and Assumption \ref{past} we have that 
\begin{equation}
    \omega_x(X)=\omega_y(X), \qquad X\in \mathcal{E}_x = \mathcal{E}_y.
\end{equation}
\end{Prp}
\begin{proof}
 We start by observing that $x\in J^+(y)$, therefore $J^+(x)\subset J^+(y)$, and $y\in J^+(x)$, therefore $J^+(y)\subset J^+(x)$, and thus by Assumption \ref{future2} we have 
\begin{equation}\label{eq:equalfuture}
    \mathcal{E}_x\subset \mathcal{E}_y \text{ and } \mathcal{E}_y\subset \mathcal{E}_x, 
\end{equation}
which can be combined to give  $\mathcal{E}_x = \mathcal{E}_y$. 

Next, we have that $x\in J^-(y)$, therefore $J^-(x)\subset J^-(y)$, and $y\in J^-(x)$, therefore $J^-(y)\subset J^-(x)$, and thus by Assumption \ref{past} we have
\begin{equation}\label{eq:equalpast1}
H\big(V^{-}_{x}(\mathcal{M})\big) =  \vec{\Pi}_{\iota \in \left[\mathfrak{I}_x(\mathcal{M})\backslash  \mathfrak{I}_y(\mathcal{M})\right] }\, \pi_{\xi_\iota}^{P_\iota}\,\, \cdot H\big(V^{-}_{y}(\mathcal{M})\big) 
\end{equation}
and 
\begin{equation}\label{eq:equalpast2}
H\big(V^{-}_{y}(\mathcal{M})\big) =  \vec{\Pi}_{\iota \in \left[\mathfrak{I}_y(\mathcal{M})\backslash  \mathfrak{I}_x(\mathcal{M})\right] }\, \pi_{\xi_\iota}^{P_\iota}\,\, \cdot H\big(V^{-}_{x}(\mathcal{M})\big).
\end{equation}
Plugging \eqref{eq:equalpast1} into the right hand side of \eqref{eq:equalpast2} and vice versa gives
\begin{align}
\vec{\Pi}_{\iota \in \left[\mathfrak{I}_x(\mathcal{M})\backslash  \mathfrak{I}_y(\mathcal{M})\right] }\, \pi_{\xi_\iota}^{P_\iota}\,\, \cdot \vec{\Pi}_{\iota \in \left[\mathfrak{I}_y(\mathcal{M})\backslash  \mathfrak{I}_x(\mathcal{M})\right] }\, \pi_{\xi_\iota}^{P_\iota}\,\, &=\\ \vec{\Pi}_{\iota \in \left[\mathfrak{I}_y(\mathcal{M})\backslash  \mathfrak{I}_x(\mathcal{M})\right] }\, \pi_{\xi_\iota}^{P_\iota}\,\,\cdot \vec{\Pi}_{\iota \in \left[\mathfrak{I}_x(\mathcal{M})\backslash  \mathfrak{I}_y(\mathcal{M})\right] }\, \pi_{\xi_\iota}^{P_\iota}\,\,&= \mathds{1}
\end{align}
observing, that both $\vec{\Pi}_{\iota \in \left[\mathfrak{I}_x(\mathcal{M})\backslash  \mathfrak{I}_y(\mathcal{M})\right] }\, \pi_{\xi_\iota}^{P_\iota}\,\, $ and $ \vec{\Pi}_{\iota \in \left[\mathfrak{I}_y(\mathcal{M})\backslash  \mathfrak{I}_x(\mathcal{M})\right] }\, \pi_{\xi_\iota}^{P_\iota}\,\,$ are products of projection operators and the identity is the only projection operator that is invertible we get that the sets $\mathfrak{I}_x(\mathcal{M})\backslash  \mathfrak{I}_y(\mathcal{M})$ and $\mathfrak{I}_y(\mathcal{M})\backslash  \mathfrak{I}_x(\mathcal{M})$ have to be empty and thus we get immediately
\begin{equation}\label{eq:equalpast}
 H\big(V^{-}_{y}(\mathcal{M})\big)=H\big(V^{-}_{x}(\mathcal{M})\big). 
\end{equation}
Finally using the general state $\omega$ in \eqref{state-propagation} completes the result. 
\end{proof}
By Proposition \ref{equal} all points that lie along the same closed causal curve share the same algebra of potentialities and the same state. Therefore not a single event will happen along a closed causal curve and as a consequence all points on these curves are physically indistinguishable. This means that there does not exist a single observation that any kind of observer could carry out to distinguish between different points along a closed causal curve. As a consequence it resolves the tension in the semi-classical setup between the causal structure in the relativistic formulation of the ETH approach to QT that does not allow for closed causal chains and the causal structure in the underlying spacetime manifold featuring closed causal curves.   \\
The result discussed in this section can be understood as a \textit{quasi passive state}. A passive state, as introduced in \cite{froehlich2019review} for the non-relativistic formulation of the ETH approach, is a state such that for any time $t>t_0$ the center of the centralizer is trivial and hence no further events occur after the time $t_0$. What was showed here is that no event can occur within any region that features closed causal curves. However, in principle one can imagine that a causal curve can leave the region where closed causal curves exist. Apriori there is nothing preventing events from happening at points along the causal curve that lie in the future but outside of the region where closed causal curves exist. Therefore the states discussed here are \textit{quasi passive} in the sense that there exists a macroscopic region where no events occur but, in principle, events can occur in the future of this region.

\section{Conclusion}\label{sec:conclusion} 
In the present paper we demonstrated that, assuming the Causal Compatibility Conjecture, points on closed causal curves are physically indistinguishable in the semi classical formulation of the ETH approach to QT and therefore closed causal curves are physically irrelevant. We thus showed that in a region where closed causal curves exist, all states are quasi passive states. 

This result is in line with the philosophy in Hawking's Chronology protection Conjecture \cite{hawking1992chronology} as it tells us that in principle we could remove the non-chronal region from the spacetime without changing anything about the physics outside of that region. Note that this is in stark contrast to the modifications to the standard formulation of QT studied in \cite{hartle1994unitarity} where the existence of closed causal curves at some point in the future already changes the physical dynamics in the present. Note also that the present result is in contrast to the model by Deutsch \cite{deutsch1991quantum} in the following way: His model assumes that an interaction occurred inside the non-chronal region between a quantum system and a past version of itself. However, for an interaction and especially for any sort of successful measurement to occur in the context of the ETH approach to QT, the occurrence of an event is required. In the present paper we showed that there can not be any events in the non-chronal region and hence the assumptions underlying \cite{deutsch1991quantum} do not apply in this context. 

(Quasi) passive states as discussed in the present paper are an interesting aspect of the ETH approach to QT as, according to Fr\"ohlich \cite{froehlich2019review}, thermal equilibrium states are passive states.  Particularly in the cosmological setting thermal (quasi) equilibrium states might play a role if, along the lines of \cite{padmanabhan2017cosmic,padmanabhan2017atoms,padmanabhan2017we} one thinks of the evolution of the universe as a sort of transition between a (quasi) thermal state with a large cosmological constant $\Lambda$ and a thermal state with a small cosmological constant. Ideas along this line have also been explored in \cite{paganini2020proposal}. Thus the properties of passive states in the ETH approach to QT should be studied in more detail.   

Finally, the argument presented here resolves the tension between the fact that in the full relativistic formulation of the ETH approach to QT closed causal chains do not exist and the closed causal curves present in the underlying spacetime in the semi classical setup. It can be expected that closed causal curves in General Relativity are a pathology stemming from the fact that the theory fails to take quantum effects into account. It thus indicates that a complete quantum theory of gravity will not admit time travel as a physical phenomenon. 

\section*{Acknowledgment} 
C.F.P. is funded by the SNSF grant  P2SKP2 178198. I would like to thank J\"urg Fr\"ohlich for a patient explanation of his ETH approach and for many interesting discussions, further I would like to thank Isha Kotecha, Markus Strehlau, Felix Finster and Marco Oppio for valuable feedback.

\bibliographystyle{amsplain}
\bibliography{claudio}

\end{document}